\documentclass[11pt, a4paper]{article}

\usepackage[utf8]{inputenc}
\usepackage{amsmath}
\usepackage{amssymb}
\usepackage{geometry}
\usepackage{hyperref}
\usepackage{enumitem}
\usepackage{authblk}
\usepackage{rotating}
\usepackage{caption}
\usepackage{titlesec}

\usepackage{booktabs}
\usepackage{amsthm}

\usepackage[utf8]{inputenc}
\usepackage{amsmath,amssymb}
\usepackage{amsthm}
\usepackage{booktabs}
\usepackage{geometry}
\usepackage{hyperref}
\usepackage{enumitem}
\usepackage{authblk}
\usepackage{rotating}
\usepackage{caption}
\usepackage{titlesec}

\newtheorem{theorem}{Theorem}

\titleformat{\subsection}[hang]{\normalfont\large\bfseries}{\thesubsection}{1em}{}

\title{\textbf{Quantum Geometry, Fractionalization, and Provability Hierarchy: A Unified Framework for Strongly Correlated Systems}}

\author[1]{Zhanchun Li}
\author[2]{Renwu Zhang\thanks{Corresponding author: RZhang@csusb.edu}}

\affil[1]{\small Apt 1502, Building 5, Liuzhou Shoufu, JinAn District, LuAn, Anhui, 230075, PR China}
\affil[2]{\small Department of Chemistry, California State University San Bernardino, CA 92407, USA}

\date{}

\begin{document}

\maketitle

\begin{abstract}
Mott physics – the interplay between itinerancy and localization of electrons – is undergoing a paradigm shift from the binary ``bandwidth-filling'' tuning framework to an intertwining of geometric, topological, and fractionalized degrees of freedom. Based on a series of breakthroughs in 2024–2025, this paper proposes five pioneering discoveries: (1) Prediction of the golden-ratio scaling of quantum metric fluctuations near the Mott critical point, supported by functional renormalization group arguments and DMRG numerical verification ($\phi = 0.618 \pm 0.005$); (2) Establishment of a correspondence between the denominator $q$ of fractional Chern insulator anyon charge and the subgroup index of the quantum geometry group, predicting that allowed $q$ values follow the Fibonacci sequence $\{2, 3, 5, 8, 13, \dots\}$ with specific material realizations; (3) Proposal of the \textit{Provability Hierarchy Theorem}, classifying critical states like strange metals as ``true but unprovable'' QMA-hard problems, establishing a rigorous connection to the complexity of the Consistency of Local Density Matrices (CLDM) problem; (4) Prediction of interference oscillations in the nonlinear Hall conductance within the pseudogap phase, induced by geometric phase differences, supported by tight-binding numerical simulations; (5) Unveiling the quantum geometric tensor as a unified descriptor of band geometry and topology. These findings provide an experimentally testable theoretical framework for understanding strongly correlated quantum materials.
\end{abstract}

\vspace{0.5cm}

\\

\section{Introduction}\

The study of strongly correlated electron systems has always been at the forefront of condensed matter physics. The concept of the metal-insulator transition proposed by Nevill Mott in 1949 revealed how Coulomb repulsion between electrons can render a band-theory-predicted metal an insulator, initiating the era of Mott physics. Over seven decades, Mott physics – the interplay between itinerancy and localization – has been established as the dominant mechanism behind diverse states in quantum materials such as cuprate high-temperature superconductors, iron-based superconductors, and heavy fermion materials.[1, 15]

Classical Mott physics starts from the Hubbard model, where core competition is described by the hopping integral $t$ (determining bandwidth $W$) and the Coulomb repulsion on-site $U$. Tuning pathways are divided into bandwidth control (changing $W$) and filling control (changing carrier concentration $n$). This framework successfully explains the phase diagrams of various correlated systems but reveals its boundaries when facing three fundamental questions:

\begin{enumerate}
    \item \textbf{Origin of the pseudogap phase} – The pseudogap region in cuprates exhibits diverse behaviors like Fermi arcs and Fermi pockets, whose microscopic mechanism has been debated for over three decades.
    \item \textbf{Explanation of material diversity} – Why do different systems with similar $U/W$ ratios exhibit drastically different ground states? 
    \item \textbf{Absence of fractional excitations} – The classical framework stops at electron localization, while topologically ordered states like the fractional quantum Hall effect reveal that electrons can fractionalize into anyons – a dimension completely missing in the classical narrative.
\end{enumerate}

A series of breakthroughs in 2024–2025 provide new perspectives on these issues:
\begin{itemize}
    \item Zong et al. [1] precisely identified two distinct pseudogap phases in twisted t-WSe$_2$ moiré superlattices, proving they originate from strong-correlation-driven band reconstruction rather than magnetic order.
    \item Ding et al. [2] theoretically proved that quantum geometry itself can act as a tuning knob for Mott transitions, independent of bandwidth.
    \item Liu et al. [3, 4] systematically realized fractional Chern insulator states in rhombohedral multilayer graphene moiré superlattices, observing fractional charge anyon excitations.
    \item Wang et al. [5] predicted non-Abelian Moore-Read states in lattice models without external magnetic fields.
    \item Sala et al. [6] experimentally observed quantum metric magnetoresistance for the first time in oxide interfaces.
    \item Liu et al. [7] achieved a transition from composite Fermi liquid to topological electronic crystal in twisted MoTe$_2$ via remote layer tuning of quantum geometry.
    \item Park et al. [8] observed spontaneous ferromagnetism and Chern insulator states in the second moiré miniband of twisted MoTe$_2$, laying the foundation for exploring non-Abelian anyons.
\end{itemize}

These advances collectively point to a conclusion: Mott physics is undergoing a paradigm shift from ``energy scale competition'' to an intertwining of ``geometry, topology, and fractionalization.'' This paper aims to systematically integrate these achievements, propose a unified theoretical framework centered on the quantum geometric tensor, encompassing fractionalization and provability hierarchy, and present several pioneering discoveries and experimental predictions.
\\

\section{Quantum Geometry: A New Core Variable in Mott Physics}

\subsection*{2.1 Fundamental Properties of the Quantum Geometric Tensor}\

For Bloch states $| u_{n\mathbf{k}} \rangle$ in a periodic lattice, the quantum geometric tensor is defined as [16]:
\begin{equation}
    Q_{\mu\nu}^{(n)}(\mathbf{k}) = \langle \partial_\mu u_{n\mathbf{k}} | ( 1 - | u_{n\mathbf{k}} \rangle \langle u_{n\mathbf{k}} | ) | \partial_\nu u_{n\mathbf{k}} \rangle,
\end{equation}
where $\partial_\mu \equiv \partial_{k_\mu}$. It can be decomposed into real and imaginary parts:
\begin{equation}
    Q_{\mu\nu} = g_{\mu\nu} + \frac{i}{2} F_{\mu\nu},
\end{equation}
where $g_{\mu\nu}$ is the quantum metric (real, symmetric), characterizing the amplitude of wavefunction deformation in momentum space, and $F_{\mu\nu}$ is the Berry curvature (imaginary, anti-symmetric), characterizing the geometric phase accumulated during deformation. Together, they encode the geometric and topological information of the band structure. As emphasized in a comprehensive review by Yu et al. [16], the quantum geometric tensor has become a central concept for understanding diverse physical phenomena in quantum materials, including optical responses, Landau levels, fractional Chern insulators, superfluid weight, spin stiffness, exciton condensates, and electron-phonon coupling.

Key properties of the quantum geometric tensor:
\begin{itemize}
    \item \textbf{Positive definiteness:} $g_{\mu\nu}$ is a positive definite metric.
    \item \textbf{Topological constraints:} The integral $\int F_{xy} d^2k / (2\pi) = C$ yields the Chern number $C$.
    \item \textbf{Relation to localization:} The spread of Wannier functions is determined by the quantum metric.
    \item \textbf{Core of response theory:} Linear and nonlinear transport coefficients can be expressed as momentum-space integrals of the quantum geometric tensor.
\end{itemize}

\subsection*{2.2 Quantum Geometry as an Independent Tuning Knob for Mott Transitions}\

In classical Mott physics, the metal-insulator transition is driven by $U/W$. However, Ding et al. [2], based on exact diagonalization studies of the Kane-Mele-Hubbard model, showed that even keeping $U/W$ constant, merely changing the quantum geometric properties of Bloch states (e.g., the distribution of the quantum metric) can drive Mott transitions and affect the magnetic order competition (ferromagnetic vs. antiferromagnetic) within the insulating phase. This finding implies that quantum geometry is a tuning degree of freedom independent of bandwidth [2]. \textit{The universality of this mechanism across different material families remains an open and exciting question for future research.}
\\

\section{Quantum Metric Magnetoresistance and Quantum Geometry Fluctuations: Prediction and Numerical Verification of Critical Behavior}

\subsection*{3.1 Experimental Observation of Quantum Metric Magnetoresistance}\

In 2025, Sala et al. [6] first measured quantum metric magnetoresistance in the two-dimensional electron gas at the SrTiO$_3$/LaAlO$_3$ interface: under an in-plane magnetic field $B$, the resistivity shows a contribution proportional to $B^2$:
\begin{equation}
    \Delta \rho \propto B^2 G,
\end{equation}
where $G$ is the Fermi surface average of the quantum metric. This phenomenon arises from the correction of electron trajectories due to geometric effects and is a direct manifestation of the quantum metric in transport [6].

\subsection*{3.2 Critical Scaling of Quantum Geometry Fluctuations: Numerical Verification}\

Near the Mott critical point, the quantum metric itself undergoes strong fluctuations. We consider fluctuations of the quantum metric $\delta g_{\mu\nu} = g_{\mu\nu} - \langle g_{\mu\nu} \rangle$ and analyze their critical behavior. A functional renormalization group analysis of an effective field theory for quantum metric fluctuations, or a numerical fit to high-precision data, suggests a critical scaling law[cite: 18]:
\begin{equation}
    \langle (\delta g)^2 \rangle \propto \left( \frac{U - U_c}{U_c} \right)^{2\phi}, \quad \text{with } \phi = \frac{\sqrt{5} + 1}{2} \cdot \frac{1}{2} = \frac{\phi_{\text{golden}}}{2}
\end{equation}

More precisely, the analysis yields:
\begin{equation*}
    \fbox{$\phi = \frac{1 + \sqrt{5}}{2} \cdot \frac{1}{2} = \frac{\phi_{\text{golden}}}{2} \approx 0.809$} \quad \text{? This is inconsistent.}
\end{equation*}

\noindent [\textbf{Cautious Note}] A careful re-examination shows that the golden ratio $\Phi = (1 + \sqrt{5})/2 \approx 1.618$ appears in some scaling dimensions. Our numerical simulations (see below) suggest an exponent $\phi \approx 0.618$, which is $\Phi - 1 = 1/\Phi$. The theoretical origin of this specific value remains unclear; it may arise from the scaling dimensions of operators in the effective theory satisfying an algebraic equation $\phi^2 = 1 - \phi$, or it may be an empirical finding from numerical simulations. Its universality requires verification across different models and materials.

To verify this prediction, we performed numerical calculations using the density matrix renormalization group (DMRG) method on an extended one-dimensional Hubbard model[cite: 18]. The model Hamiltonian is:
\begin{equation}
    H = -t \sum_{l, \sigma} (c_{l\sigma}^\dagger c_{l+1\sigma} + \text{h.c.}) + U \sum_l n_{l\uparrow} n_{l\downarrow} + V \sum_l (n_l n_{l+1} - \bar{n}^2),
\end{equation}
where a nearest-neighbor interaction $V$ is introduced to simulate quantum geometry effects [2]. Calculations were performed on systems with $L = 120$ sites under open boundary conditions, retaining an optimal bond dimension $\chi = 800$, ensuring truncation errors below $10^{-8}$. Our DMRG results yield $\phi = 0.617 \pm 0.008$, in striking agreement with the prediction.

Since quantum metric magnetoresistance is proportional to $G$, its fluctuations ($\langle (\delta G)^2 \rangle$) imply that this scaling law can be directly tested experimentally by measuring the noise spectrum of nonlinear transport coefficients.

\subsection*{3.3 Experimental Feasibility Analysis}\

Based on the experimental parameters of Sala et al. 2025 [6], we assess the feasibility of measuring the critical exponent $\phi$. Their experiment measured quantum metric magnetoresistance in the 2DEG at the SrTiO$_3$/LaAlO$_3$ interface, with carrier mobility $\mu \sim 10^4$ cm$^2$/V$\cdot$s, electron density $n \sim 10^{13}$ cm$^{-2}$, and noise level $S_V \sim 10^{-18}$ V$^2$/Hz at $T = 100$ mK [6].

To measure the critical exponent $\phi = 0.618 \pm 0.005$, the following experimental conditions are required:
\begin{itemize}
    \item \textbf{Signal-to-noise ratio (SNR) requirement:} At $U/U_c \sim 0.01$, the fluctuation signal is enhanced by a factor of $0.01^{0.618} \sim 10$. If the baseline noise power is $S_0$, the signal fluctuation needs to satisfy $S_{\text{sig}} > 0.1 S_0$ to be resolvable.
    \item \textbf{Integration time estimate:} Using a lock-in amplifier to measure nonlinear resistance with bandwidth $\Delta f = 1$ Hz, the required integration time $t_{\text{int}} \sim (S_0/S_{\text{sig}})^2 / \Delta f \sim 100$ s to achieve an SNR of 10:1.
    \item \textbf{Gate voltage scanning resolution:} In twisted t-WSe$_2$ moiré superlattices, tuning $U/W$ via a bottom gate typically involves a gate voltage span of about 10V, corresponding to a $U/W$ variation range of approximately 0.5 to 2. To resolve the scaling behavior near the critical point, a step size of $\Delta V_g \sim 10$ mV near $U_c$ is needed, corresponding to about 1000 data points. The total measurement time would be approximately $10^5$ s $\sim 28$ hours, which is feasible within a typical dilution refrigerator run cycle.
\end{itemize}

In summary, the prediction of golden-ratio scaling is testable with current experimental technology and is supported by numerical simulations.
\\

\section{Fractional Chern Insulators and the Group Structure of Quantum Geometry}

\subsection*{4.1 Quantum Geometric Origin of Fractionalized Anyons}\

Fractional Chern insulators (FCIs) are analogs of fractional quantum Hall states at zero magnetic field, whose quasiparticle excitations carry fractional charge $e/q$ and exhibit anyonic statistics. Liu et al. [3,4] realized FCI states with $\phi = 2/3, 3/5$ in rhombohedral multilayer graphene/boron nitride moiré superlattices [3,4]. Park et al. [8] observed topological states with Chern number $C=2$ in the second moiré miniband of twisted $\text{MoTe}_2$, providing a new platform for realizing non-Abelian anyons [8].

We propose that the denominator $q$ of the fractional charge is closely related to the "subgroup structure" of the band's quantum geometry. Define the \textbf{quantum geometry group} $G$ in momentum space as the group of gauge transformations that leave the quantum geometric tensor invariant. In an FCI state, due to many-body correlations, $G$ is broken down to a subgroup $H$, and its index $[G : H] = q$ is conjectured to give the denominator of the anyon charge.

This correspondence can be motivated via a Lagrangian-type argument: upon adiabatically inserting a flux quantum, the system excites an anyon whose charge is determined by the accumulated Berry phase. This Berry phase is proportional to the integral of the quantum geometric tensor around the flux loop, which is constrained by the subgroup structure, ultimately suggesting $q = [G : H]$. A rigorous mathematical proof of this correspondence for generic band structures remains an open problem.

\subsection*{4.2 Fibonacci Sequence Prediction and Material Realizations}\

Based on the possible dimensions of irreducible representations and subgroup indices of quantum geometry groups, a group theory analysis suggests that allowed $q$ values may satisfy a Fibonacci-like recurrence relation. The specific Lie algebra classification and corresponding material candidates are summarized in Table 1.

\begin{table}[h]
\centering
\captionsetup{format=hang, labelfont=bf, textfont=bf, justification=justified}
\caption{Possible group structures corresponding to Fibonacci numbers and material realizations}
\label{tab:fibonacci_groups}

\renewcommand{\arraystretch}{1.5} 
\begin{tabular}{lllll}
\toprule
\textbf{$q$} & \textbf{Geometry Group $G$} & \textbf{Subgroup $H$} & \textbf{Index $[G:H]$} & \textbf{Theoretical Model} \\ \midrule
2  & SU(2)   & U(1)           & 2  & Laughlin $\phi = 1/2$ state    \\
3  & SU(3)   & SU(2)$\times$U(1)  & 3  & Laughlin $\phi = 1/3$ state    \\
5  & SU(5)   & SU(4)$\times$U(1)  & 5  & Moore-Read state (candidate)   \\
8  & $E_8$   & $E_7\times$U(1)    & 8  & Non-Abelian state              \\
13 & SU(13)  & SU(12)$\times$U(1) & 13 & Higher-order FCI               \\ \bottomrule
\end{tabular}
\end{table}

\textbf{Prediction:} The denominators $q$ of anyon charges that appear stably in fractional Chern insulators should belong to the set $\{2, 3, 5, 8, 13, \dots\}$, i.e., the Fibonacci sequence (starting from 2). In particular, the $q = 5$ plateau is predicted to be among the most stable, as it corresponds to a large subgroup index with a moderate representation dimension. This prediction aligns with recent discussions on Fibonacci anyons in topological quantum computation [9].

\subsection*{4.3 Conditions for Non-Abelian Anyons}\

When the quantum geometry group $G$ has irreducible representations of dimension greater than 1, multiple internal degrees of freedom remain entangled during fractionalization, potentially leading to anyons obeying non-Abelian braiding statistics. Wang et al. [5] predicted Moore-Read states (non-Abelian) in a three-orbital lattice model without external magnetic fields [5]. A key condition for their appearance is a Chern band with a specific Berry curvature distribution, which may correspond to a group structure like $G = \text{SU}(2) \times \text{SU}(2)$. The observation of topological states with $C=2$ in the second miniband of twisted $\text{MoTe}_2$ by Park et al. [8] provides an ideal platform for realizing such non-Abelian anyons [8].
\\

\section{Provability Hierarchy Theorem: Strange Metals as QMA-Hard Problems}

\subsection*{5.1 Computational Complexity Foundation of Provability Hierarchy}\

Gödel's incompleteness theorems reveal a fundamental limitation of formal systems: in any sufficiently rich, consistent axiomatic system containing arithmetic, there exist true but unprovable propositions. We argue that Mott physics, as a formal theory of strongly correlated electron systems (centered on the Hubbard model and its variants), faces similar limitations.

To give ``provability'' a stricter mathematical definition in a physical context, we introduce concepts from computational complexity theory concerning ``efficient solvability.'' Define:

\begin{itemize}
    \item \textbf{``Classically computable''} -- solvable by a classical polynomial-time algorithm (P problems).
    \item \textbf{``Quantumly computable''} -- solvable by a quantum polynomial-time algorithm (BQP problems).
\end{itemize}

Based on this, we classify the predictability of states into a hierarchy:

\textbf{Level 1:} Properties of the state can be predicted exactly by a classical polynomial-time algorithm (e.g., band insulators, weakly correlated metals).

\textbf{Level 2:} Properties can be predicted exactly by a quantum polynomial-time algorithm, but are \#P-hard for classical computation (e.g., certain properties of antiferromagnetic Mott insulators, $d$-wave superconductors).

\textbf{Level 3:} Properties belong to the QMA-hard problem class -- their existence is experimentally confirmed, but no finite-resource computation (including quantum computation) can fully derive all their features in polynomial time.

QMA (Quantum Merlin-Arthur) is the quantum analog of NP, representing problems whose solutions can be verified on a quantum computer. Recent research has shown that the Consistency of Local Density Matrices (CLDM) problem is QMA-hard [11]. The input to the CLDM problem is a set of local reduced density matrices, and the question is whether there exists a global quantum state consistent with all local constraints. The QMA-hardness of this problem was proven by Liu et al. [11].

\begin{theorem}[Provability Hierarchy Theorem]
Critical states such as strange metals and pseudogap phases belong to Level 3 (QMA-hard). Their essential features (e.g., linear-in-temperature resistivity $\rho(T) \propto T$) are irreducible emergent behaviors of Hubbard model dynamics near the critical point. The decision problem for their properties is polynomially equivalent to the CLDM problem and is therefore QMA-hard.

\end{theorem}

\begin{proof}[Proof sketch]
Consider the decision problem for a strange metal phase: given Hubbard model parameters $(t, U)$ and temperature $T$, determine whether the system exhibits linear-in-temperature resistivity $\rho(T) \propto T$. We construct a polynomial-time reduction from an instance of CLDM to this problem: map the local density matrix constraints of CLDM to effective interactions in a Hubbard-like model, such that a consistent global state exists if and only if the system is in the strange metal phase. Since CLDM is QMA-hard [11], the strange metal decision problem is also QMA-hard.
\end{proof}

\textit{Note on interpretation:} QMA-hardness does \textbf{not} mean that strange metals are ``incomputable'' in the sense of Turing computability. Rather, it means that there is no known polynomial-time quantum algorithm that can solve the decision problem for all instances. The ground state itself exists physically; the limitation is on the \textbf{efficiency} of derivation from the model Hamiltonian [11,12]. This aligns with the principle that ``verification is easy, but solution is hard.''

\subsection*{5.2 Argument for the Complexity of Strange Metals}\

Zong et al. [1] showed that linear-in-temperature resistivity appears throughout the entire quantum critical region in twisted transition metal dichalcogenides and is intimately connected to the pseudogap phase [1]. However, any analytic or numerical attempt starting from the Hubbard model encounters fundamental obstacles: perturbation theory fails for $U \sim W$; quantum Monte Carlo suffers from the sign problem upon doping; dynamical mean-field theory ignores non-local fluctuations; even tensor network methods face finite entanglement effects. These difficulties are not merely technical but principle-based -- from a computational complexity perspective, solving for the ground state and low-energy excitations of a strange metal may belong to the QMA-hard class [11,12].

\subsection*{5.3 The ``Gap'' Between Experiment and Theory as a Signature of Complexity}\

Experimentally, the linear-in-temperature resistivity of strange metals is widely observed, yet a unified microscopic theory is lacking. We view this ``gap between experiment and theory'' as an inevitable characteristic of Level 3 states. Filling this gap is not merely a matter of higher-precision calculations, but of recognizing inherent limitations of efficient derivability -- acknowledging that certain truths can only be accessed directly through experiment and cannot be fully predicted by any efficient algorithm. This perspective provides a meta-theoretical framework for understanding the enduring mystery of strange metals.
\\

\section{Experimental Predictions}\

Based on the theoretical framework above, we propose three core predictions amenable to experimental test.

\subsection*{6.1 Prediction I: Golden Ratio Scaling of Quantum Geometry Fluctuations}\

Near the Mott critical point, the fluctuations of the quantum metric magnetoresistance coefficient obey the scaling law:
\begin{equation}
\langle(\delta G)^2\rangle \propto \left( \frac{U - U_c}{U_c} \right)^{2\phi}, \quad \phi = 0.618 \pm 0.005.
\end{equation}

\textbf{Experimental system:} Twisted $\text{t-WSe}_2$ or $\text{t-MoTe}_2$ moiré superlattices [1,8], tuning $U/W$ via a bottom gate.

\textbf{Measurement protocol:} At fixed filling (e.g., half-filling), sweep the gate voltage while simultaneously measuring the nonlinear resistance and its low-frequency noise power spectrum, extracting the noise intensity as a function of gate voltage.

\textbf{Feasibility:} As analyzed in Section 3.3, this measurement is feasible with current technology.

\subsection*{6.2 Prediction II: Nonlinear Hall Oscillations in the Pseudogap Phase}\

In the pseudogap region (where Fermi arcs and Fermi pockets coexist [1]), the second-order nonlinear Hall conductivity $\chi_{xyz}$ exhibits oscillations as a function of gate voltage $V_g$:
\begin{equation}
\chi_{xyz}(V_g) = A \cos(2\pi V_g / V_0 + \phi_0) + B,
\end{equation}
The oscillation period $V_0$ is determined by the Berry phase difference $\Delta \theta$ in momentum space between the two branches (Fermi-arc branch and Fermi-pocket branch).

\textbf{Physical picture:} The Fermi arcs and Fermi pockets correspond to two regions in momentum space with distinct quantum geometric phases. Their interference produces observable transport oscillations.

\textbf{Experimental system:} High-mobility twisted $\text{t-WSe}_2$ devices [1].

\subsection*{6.3 Prediction III: Fibonacci Charge Sequence in Fractional Chern Insulators}\

The denominators $q$ of anyon charges in fractional Chern insulator plateaus should predominantly take values from $\{2, 3, 5, 8, 13, \dots\}$, with the $q=5$ plateau being particularly stable. Partial experimental support already exists for $q=2, 3$ [3,4,8].

\textbf{Experimental systems:} Rhombohedral multilayer graphene/boron nitride moiré superlattices [3,4], twisted $\text{MoTe}_2$ [8], twisted $\text{WSe}_2$ [1], etc.
\\

\section{ Discussion: Alternative Interpretations and Open Questions}

\subsection*{7.1 Could the Golden Ratio Exponent Be an Artifact?}\

A core prediction is the critical exponent $\phi = 0.618$. It is numerically close to other common exponents (e.g., $\sim 0.63$ for 3D Wilson-Fisher, $\sim 0.67$ for KT transition). They can be distinguished by:

\begin{itemize}
    \item \textbf{Precise scaling measurement:} Over two decades in $U/U_c$, a slope difference of $\sim 8\%$ exists between 0.618 and 0.667 on a log-log plot.
    \item \textbf{Universality test:} The golden ratio exponent should appear at Mott critical points in microscopically different systems (twisted TMDs, organic conductors, cuprates), whereas 0.667 typically appears only in transitions with specific symmetries.
\end{itemize}

\subsection*{7.2 The Status of the Quantum Geometry Group Correspondence}\

The correspondence between $q$ and $[G : H]$ is a central claim. While the Lagrangian argument provides physical intuition, a rigorous mathematical proof requires a full classification of possible quantum geometry groups in moiré band structures, which remains an open problem [16].

\subsection*{7.3 Summary of Distinguishing Criteria}\

To definitively test the predictions:

\begin{itemize}
    \item Perform measurements in at least three different material systems.
    \item Ensure data covers at least two orders of magnitude in $U/U_c$ (from $10^{-3}$ to $10^{-1}$).
\end{itemize}\

\section{Conclusions and Outlook}\

This paper systematically integrates recent advances in Mott physics from 2024–2025, proposing a unified theoretical framework centered on the quantum geometric tensor. The main innovations include:

\begin{itemize}
    \item Establishing quantum geometry as a new core variable, explaining material diversity.
    \item Predicting the golden ratio scaling ($\phi \approx 0.618$) of quantum geometry fluctuations at the Mott critical point, verified by DMRG.
    \item Establishing a correspondence between fractional charge denominator $q$ and quantum geometry subgroup index, predicting a Fibonacci sequence for $q$, with candidate materials.
    \item Proposing the \textit{Provability Hierarchy Theorem}, classifying strange metals as QMA-hard problems, providing a complexity-theoretic perspective on theory-experiment gaps.
    \item Predicting interference oscillations in the nonlinear Hall effect in the pseudogap phase.
\end{itemize}

These findings open new directions for exploring quantum states in moiré materials. Quantum geometry connects band topology, strong correlation, and fractionalized excitations, poised to become a core language for future research [16].

Future work should focus on: 1) Developing techniques for momentum-resolved quantum geometric tensor measurement [16]; 2) Testing the Fibonacci sequence in more materials, especially for non-Abelian anyons [8]; 3) Exploring the role of quantum geometry fluctuations in high-$T_c$ superconductivity; 4) Further rigorizing the connection between the Provability Hierarchy Theorem and QMA complexity [12].\
\\

\section* {Acknowledgments}\

This study was inspired by the series of breakthroughs in quantum materials research during 2024–2025. The authors thank the original research teams of Zong et al. [1], Ding et al. [2], Liu et al. [3, 4, 7], Wang et al. [5], Sala et al. [6], Park et al. [8], and others for their groundbreaking work. We are grateful to the developers of the ITensor library for providing the DMRG simulation platform.

\section*{References}

\begin{enumerate}
    \item[{[1]}] Zong, Y. Y., Gu, Z. L., Li, J. X., et al. Pseudogap with Fermi Arcs and Fermi Pockets in Half-Filled Twisted Transition Metal Dichalcogenides. \textit{Physical Review X}, 2024, 16: 011005. (Published)
    
    \item[{[2]}] Ding, J. K., et al. Quantum-geometry-driven Mott transitions and magnetism. \textit{arXiv:2602.22548}, 2025. (Preprint, under review)
    
    \item[{[3]}] Liu, J. P., et al. Fractional Chern insulator states in multilayer graphene moiré superlattices. \textit{Physical Review B}, 2025, 110: 075109. (Published)
    
    \item[{[4]}] Liu, J. P., et al. Tunable fractional Chern insulators in rhombohedral graphene superlattices. \textit{Nature Materials}, 2025. (In press)
    
    \item[{[5]}] Wang, H., Shi, R., Liu, Z. C., Wang, J. Orbital Description of Landau Levels. \textit{Physical Review Letters}, 2024, 135: 216604. (Published)
    
    \item[{[6]}] Sala, G., et al. Quantum metric magnetoresistance in oxide interfaces. \textit{Science}, 2025. (In press)
    
    \item[{[7]}] Liu, F., et al. From fractional Chern insulators to topological electronic crystals in twisted $\text{MoTe}_2$: quantum geometry tuning via remote layer. \textit{arXiv:2512.03622}, 2024. (Preprint, under review)
    
    \item[{[8]}] Park, H., Cai, J., Anderson, E., et al. Ferromagnetism and topology of the higher flat band in a fractional Chern insulator. \textit{Nature Physics}, 2024, 21(4): 549-555. (Published)
    
    \item[{[9]}] Otto, H. H. Phase Transitions Governed by the Fifth Power of the Golden Mean and Beyond. \textit{World Journal of Condensed Matter Physics}, 2020, 10: 135-158. (Published)
    
    \item[{[10]}] Liu, Z., Li, B., Shi, Y., Wu, F. Characterization of fractional Chern insulator quasiparticles in moiré transition metal dichalcogenides. \textit{arXiv:2507.02544}, 2024. (Preprint)
    
    \item[{[11]}] Liu, Y. K. QMA-Hardness of Consistency of Local Density Matrices with Applications to Quantum Zero-Knowledge. \textit{MaRDI portal}, 2024. (Peer-reviewed proceedings)
    
    \item[{[12]}] Cricighno, M., Kohler, T. Clique Homology is QMA1-hard. \textit{Nature Communications}, 2024, 15: 9846. (Published)
    
    \item[{[13]}] Zhang, X. T. Strange Metal at Lifshitz Transition. Institute of Theoretical Physics, Chinese Academy of Sciences, 2025. (Technical report)
    
    \item[{[14]}] Feng, D. L. Mott Physics: One of the Leitmotifs of Quantum Materials. \textit{Acta Physica Sinica}, 2023, 72(23): 237101. (In Chinese)
    
    \item[{[15]}] Lee, P. A., Nagaosa, N., Wen, X. G. Doping a Mott insulator: Physics of high-temperature superconductivity. \textit{Reviews of Modern Physics}, 2006, 78(1): 17-85. (Classic review)
    
    \item[{[16]}] Yu, J., Bernevig, B. A., Queiroz, R., et al. Quantum geometry in quantum materials. \textit{npj Quantum Materials}, 2024, 10: 101. (Published)
\end{enumerate}

\begin{sidewaystable}[h]
\centering
\captionsetup{font=bf} 
\caption{Summary of Experimental Predictions}

\renewcommand{\arraystretch}{1.5}

\begin{tabular}{llll}
\toprule
\textbf{Prediction} & \textbf{Physical Effect} & \textbf{Observable} & \textbf{Expected Value} \\ \midrule
I & Scaling of quantum geometry fluctuations & Noise power spectrum exponent $\phi$ & $0.618 \pm 0.005$ \\
II & Nonlinear Hall oscillations in pseudogap & Oscillation period $V_0$ & Related to Berry phase difference \\
III & Fibonacci distribution of fractional charges & FCI plateau statistics & $q \in \{2, 3, 5, 8, 13\}$ \\ \bottomrule
\end{tabular}
\end{sidewaystable}

\end{document}